\documentclass[twocolumn]{autart}    
\usepackage{amsmath,bm}
\usepackage{color}
\usepackage{multirow}
\usepackage{mathrsfs}
\usepackage{dsfont}
\usepackage{amssymb}
\usepackage{setspace}
\usepackage{graphicx}
\usepackage{float}

\newtheorem{remark}{Remark}

\newtheorem{assumption}{Assumption}
\newtheorem{corollary}{Corollary}
\newtheorem{proposition}{Proposition}

\newtheorem{proof}{\textbf{Proof}}

\def\begcen{\begin{center}}
\def\endcen{\end{center}}

\newcommand{\col}{\mbox{col}}

\def\diag{\mbox{diag}}

\def\calc{\mathcal{C}}
\def\calq{\mathcal{Q}}

\def\cala{\mathcal{A}}

\def\calm{\mathcal{M}}
\def\calr{\mathcal{R}}
\def\calj{\mathcal{J}}

\def\calz{\mathcal{Z}}
\def\calo{\mathcal{O}}
\def\cale{\mathcal{E}}

\def\cl{{\tt cl}}

\def\frakh{\mathcal{H}}
\def\frakj{\mathcal{J}}
\def\frakr{\mathcal{R}}
\def\bfh{{\bf H }}
\def\bfj{{\bf J}}
\def\bfr{{\bf R}}
\def\bbj{\mathbb{J}}

\def\liminf{\lim_{t \to \infty}}

\def\L2e{{\cal L}_{2e}}

\def\rea{\mathbb{R}}

\def\diag{\mbox{diag}}

\def\col{\mbox{col}}

\def\et{\varepsilon_t}
\def\diag{\mbox{diag}}

\def\min{{\mbox{min}}}
\def\max{{\mbox{max}}}

\def\et{\epsilon_t}

\def\begequarr{\begin{eqnarray}}
\def\endequarr{\end{eqnarray}}
\def\begequarrs{\begin{eqnarray*}}
\def\endequarrs{\end{eqnarray*}}
\def\begarr{\begin{array}}
\def\endarr{\end{array}}
\def\begequ{\begin{equation}}
\def\endequ{\end{equation}}
\def\lab{\label}
\def\begdes{\begin{description}}
\def\enddes{\end{description}}
\def\begenu{\begin{enumerate}}
\def\begite{\begin{itemize}}
\def\endite{\end{itemize}}
\def\endenu{\end{enumerate}}

\def\lef[{\left[\begin{array}}
\def\rig]{\end{array}\right]}
\def\qed{\hfill$\Box \Box \Box$}
\def\begcen{\begin{center}}
\def\endcen{\end{center}}
\def\begrem{\begin{remark}\rm}
\def\endrem{\end{remark}}

\def\IJRNLC{{\it Int. J. on Robust and Nonlinear Control}}
\def\TAC{{\it IEEE Trans. Automatic Control}}

\def\EJC{{\it European Journal of Control}}

\def\CDC{{\it IEEE Conference on Decision and Control}}
\def\SCL{{\it Systems \& Control Letters}}
\def\AUT{{\it Automatica}}

\def\CSM{{\it IEEE Control Systems Magazine}}

\def\begmat#1{\begin{bmatrix}#1\end{bmatrix}}
\def\begali#1{\begin{align}{#1}\end{align}}
\def\begalis#1{\begin{align*}{#1}\end{align*}}

\usepackage{color}
\newcommand{\blue}[1]{{\color{blue} #1}}

\begin{document}

\begin{frontmatter}

\title{Orbital Stabilization of Nonlinear Systems via Mexican Sombrero Energy Shaping and Pumping-and-Damping Injection\thanksref{footnoteinfo}} 

\thanks[footnoteinfo]{Corresponding author: W. Zhang.}
\author[SJTU,L2S]{Bowen Yi}\ead{yibowen@ymail.com},               
\author[L2S,ITMO]{Romeo Ortega}\ead{ortega@lss.supelec.fr},  
\author[L2S]{Dongjun Wu}\ead{dongjun.wu@l2s.centralesupelec.fr},
\author[SJTU]{Weidong Zhang}\ead{wdzhang@sjtu.edu.cn}

\address[SJTU]{Department of Automation, Shanghai Jiao Tong University, 200240 Shanghai, China}             
\address[L2S]{ Laboratoire des Signaux et Syst\`emes, CNRS-CentraleSup\'elec, 91192 Gif-sur-Yvette, France}
\address[ITMO]{Faculty of Control Systems and Robotics, ITMO University, St. Petersburg 197101, Russia}

\begin{keyword}                           
passivity-based control, nonlinear systems, orbital stabilisation, oscillator.         
\end{keyword}                             

\begin{abstract}                          
In this paper we show that a slight modification to the widely popular interconnection and damping assignment passivity-based control method---originally proposed for stabilization of equilibria of nonlinear systems---allows us to provide a solution to the more challenging orbital stabilization problem. Two different, though related, ways how this procedure can be applied are proposed. First, the assignment of an energy function that has a minimum in a closed curve, {\em i.e.}, with the shape of a Mexican sombrero. Second, the use of a damping matrix that changes ``sign" according to the position of the state trajectory relative to the desired orbit, that is, pumping or dissipating energy. The proposed methodologies are illustrated with the example of the induction motor and prove that it yields the industry standard field oriented control.
\end{abstract}
\end{frontmatter}

%
\section{Introduction}
\label{sec1}
%
In many practical tasks the system under control is required to operate along periodic motions, \emph{i.e.}, walking and running robots, path following, rotating electromechanical systems, AC or resonant power converters, and oscillation mechanisms in biology. As clearly explained in \cite[Section 8.4]{KHA} the stability analysis of these behaviors can be recast as a standard equilibrium stabilization problem, but this leads to very conservative results. It is more convenient, instead, to invoke the notion of stability of an invariant set, where the latter is the closed orbit associated to the periodic solution. This approach leads to the important notion of orbital stability  \cite[Definition 8.2]{KHA}.

A large number of papers and books have been devoted to analysis of orbital stability of a given dynamical system, see {\em e.g.}, \cite{CHE,FRAPOG,GUKHOL}. However, there are only a few \emph{constructive tools} available to solve the task of orbital stabilization of a controlled system. A popular approach to address this question is the virtual holonomic constraints (VHC) method, which has been tailored for mechanical systems of co-dimension one \cite{MAGCONtac,MOHetalaut,SHIetaltac,SHIetaltac10}. In the VHC method a certain subspace of the state-space is rendered attractive and invariant, leading to a projected dynamics that behaves as oscillators. This is a particular case of the framework adopted in the immersion and invariance (I\&I) technique, first reported for equilibrium stabilization in \cite{ASTORT}, and later extended for observer design and adaptive control in \cite{ASTbook}. In \cite{ORTijrnlc} it has recently been shown that I\&I can also be adapted for orbital stabilization, leading to a procedure that contains,  as particular case, the VHC designs. The only modification done to the standard I\&I technique is in the definition of the target dynamics that now should be chosen possessing periodic orbits, instead of an equilibrium at the desired point. A main drawback in both the VHC and I\&I methods is that the steady-state behavior cannot be fixed {\em a priori}, but depends on the initial states, see  \cite[Remark 2]{ORTijrnlc} for a discussion on this matter.

An alternative approach to generate oscillations is reported in \cite{STASEPtac}, where it is proposed to construct passive oscillators for Lure dynamical systems using ``sign-indefinite" feedback static mappings, which is a mechanism similar to the pumping-and-damping injection discussed below. Unfortunately, since the analysis is carried out applying the center manifold theory---that is a local notion---the obtained oscillators are assumed to have small amplitudes. Orbital stabilization designs, for some particular controlled plants, have also been reported in \cite{ARCetaltac,FRAPOG,SPO}.

The aim of this paper is to show that the widely popular interconnection and damping assignment passivity based control (IDA-PBC), originally proposed in \cite{ORTetaltac,ORTetalaut,ORTGAR} for stabilization of equilibria, can be easily be adapted to address the problem of orbital stabilization of general nonlinear systems.  This leads to two new \emph{constructive} solutions for this problem that---as usual in PBC---have a clear interpretations from the energy viewpoint. First, the assignment of an energy function that has a minimum in a closed curve, {\em i.e.}, with the shape of a Mexican sombrero. Second, the use of a damping matrix that changes ``sign" according to the position of the state trajectory relative to the desired orbit, that is, pumping or dissipating energy. As usual in all constructive nonlinear controller designs, the success of the proposed methods hinges upon our ability to solve a partial differential equation (PDE).

The remaining of the paper is organized as follows. Section \ref{sec2} revisits the standard IDA-PBC. Section \ref{sec3} introduces the problem formulation of orbital stabilization, followed by the constructive main results in Section \ref{sec4}. The application to the induction motor (IM) is reported in Section \ref{sec5}. Interestingly, we prove that the resulting controller exactly coincides with the industry standard direct field-oriented control (FOC) first proposed in \cite{BLA}. In Section \ref{sec6} the orbital stabilization of pendula is studied. The paper is wrapped-up with conclusions and future work in Section \ref{sec7}.
\\

\noindent {\em Notation.} $\mathbb{S}$ denotes the unit circle. Given a set $\cala\subset \rea^n$ and a vector $x \in \rea^n$, we denote $\|x\|_\cala:= \inf_{y\in\cala}|x-y|$, with $|x|^2:=x^\top  x$, and $B_{\varepsilon}(\cala) :=\{x\in \rea^n | \|x\|_\cala \le \varepsilon\}$. All mappings are assumed smooth. For a full-rank mapping $g(x) \in \rea^{n\times m}$ with $(m < n)$, we denote the generalized inverse as $g^\dagger(x) := [g^\top(x)g(x)]^{-1} g^\top(x)$, and $g^\bot(x) \in \rea^{(n-m) \times n}$ a full-rank left annihilator of $g(x)$.  We define the gradient transpose operator as $\nabla_x:=(\partial /\partial x)^\top$. When clear from the context the arguments of the mappings and the operator $\nabla$ are omitted.

\noindent {\em Caveat.} An abridged version of this paper will be presented in \cite{YIetalcdc}.

%
\section{Background on IDA-PBC}
\label{sec2}
%
We consider in the paper systems written in the form
\begequ
\label{fgsys}
\dot{x} = f(x) + g(x)u,
\endequ
with the state $x \in \rea^n$ and the control $u \in \rea^m,\;m \leq n$ and $g(x)$ full rank. To solve the orbital stabilization problem we propose in the paper a variation of the IDA-PBC method \cite{ORTetalaut}, normally used for regulation tasks. The objective in IDA-PBC is to find a feedback control law $u=\hat u(x)$ such that the closed-loop dynamics takes a {port-Hamitonian (pH)} form, that is,
\begequ
\label{matobj}
f(x)+g(x)\hat u(x) = \big[\calj (x) - \calr(x)\big] \nabla H(x)=:f_{\tt cl}(x)
\endequ
with $H(x) \in \rea$ the desired Hamiltonian and
\begequ
\lab{jdrd}
\calj(x) = - \calj^\top(x),\; \calr(x) = \calr^\top(x)
\endequ
the desired $(n \times n)$ interconnection and damping matrices, respectively. The matching objective \eqref{matobj} is achieved if and only if the following PDE (in $H(x)$) is solved
\begali{
\lab{pde}
  g^\bot(x)f(x)  &=  g^\bot(x) \big[\calj(x) - \calr(x)\big]\nabla H(x).
}
If this is the case, the control law is given as
\begali{
  \hat u(x) & = g^\dagger(x) \big[ \big(\calj(x) - \calr(x)\big) \nabla H(x) - f(x) \big].
\label{ida}
}

In {\em regulation} tasks, $H(x)$ has a unique minimum at the desired equilibrium and we choose the matrix $\calr(x)$ to be positive semi-definite to inject the damping required to drive the trajectory towards the equilibrium. In this paper we show that, for {\em orbital stabilization} we select $H(x)$ to have a minimum at the {\em desired orbit}---see Fig. \ref{fig:somb}. We will refer to this controller as Mexican sombrero energy assignment (MSEA) PBC.

\begin{figure}[]
  \centering
  \includegraphics[width=2.5cm]{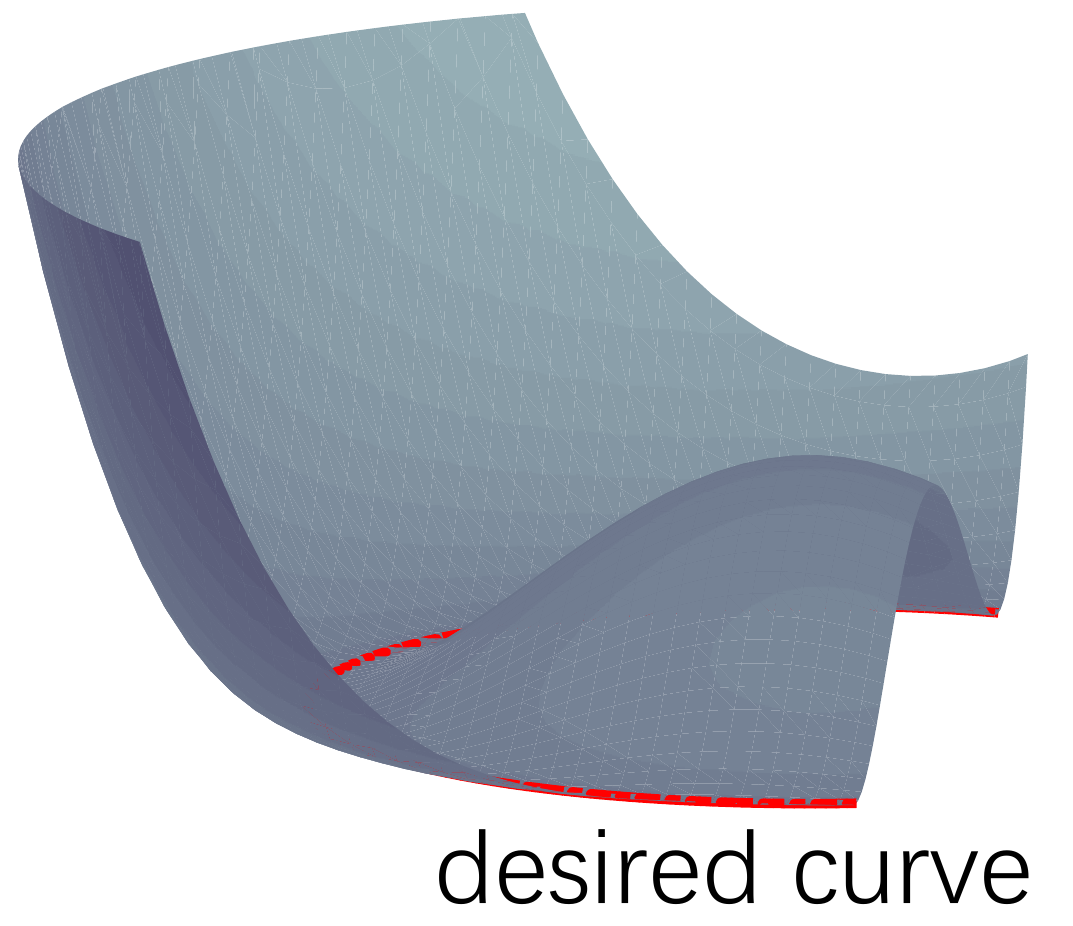}
  \quad~~
  \includegraphics[width=3.4cm]{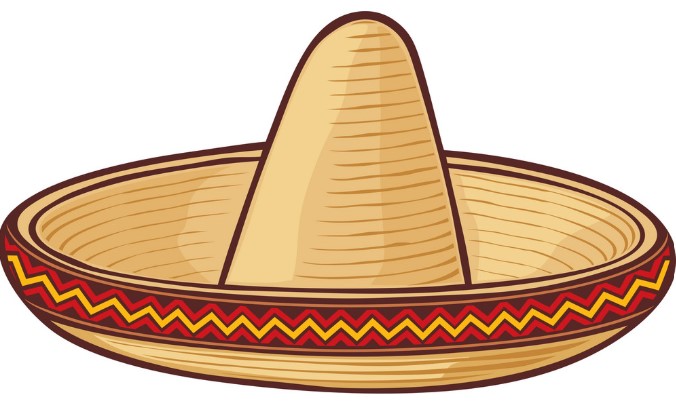}
  \caption{The closed-loop Hamiltonian with the target orbit and the Mexican sombrero in MSEA-PBC}\label{fig:somb}
\end{figure}

An alternative option is to select the ``sign" of $\calr(x)$ to {\em pump energy} or {\em inject damping} according to the relative position of the state with respect to the desired orbit---this method is called energy pumping-and-damping (EPD)-PBC. A visual illustration is given in Fig. \ref{fig:pd}.

\begin{figure}[]
  \centering
  \includegraphics[width=2.5cm]{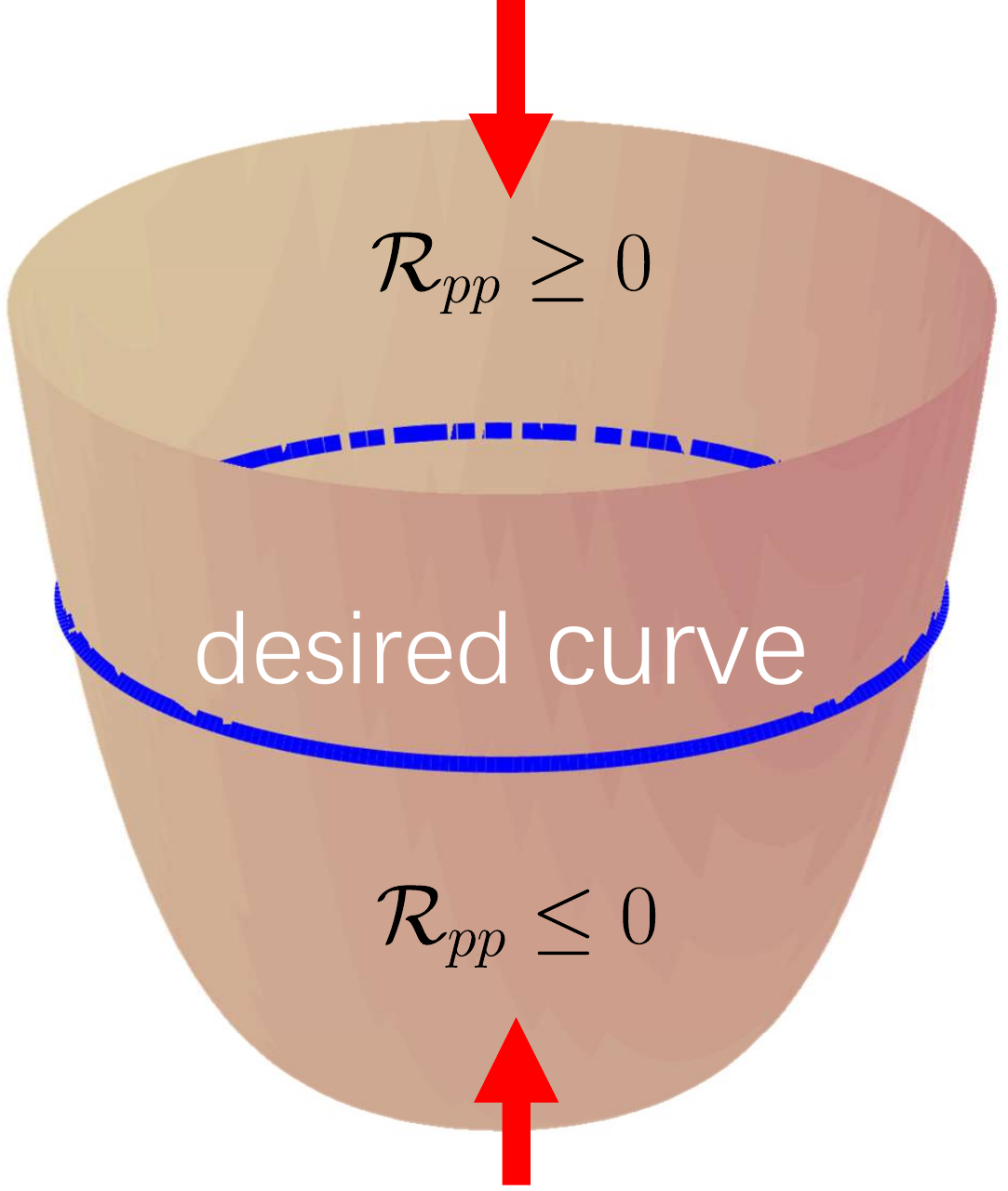}
  \caption{The desired Hamiltonian and the target orbit in EPD-PBC.}\label{fig:pd}
\end{figure}

%
\section{Problem Formulation}
\label{sec3}
%
We are interested in the paper in the generation, via IDA-PBC, of periodic solutions $X:\rea_+ \to \rea^n$ which are {\em asymptotically orbitally stable}. That is,
\begalis{
\dot X(t)  =f_{\tt cl}(X(t)) , \quad
X(t)  =X(t+T),\; \forall t \geq 0,
}
where the closed-loop vector field $f_{\tt cl}(x)$ is given in \eqref{matobj}, and the set is defined by its associated closed orbit
$$
\cala:=\{x \in \rea^n\;|\; x=X(t),\;0 \leq t \leq T\},
$$
is attractive.

We consider a particular case of periodic motion, which is defined as follows. First, split the state as $x:= \col(x_p, x_\ell)$, with $x_p \in \rea^{2},\;x_\ell \in \rea^{n-2}$. We also partition the matrices $\calj(x)$ and $\calr(x)$ conformally as
$$
\begmat{ {(\cdot)_{pp} }& (\cdot)_{p\ell} \\ (\cdot)_{\ell p} & {(\cdot)_{\ell \ell}} } \sim \begmat{ \rea^{2 \times 2} & \rea^{2 \times (n-2)} \\ \rea^{(n-2) \times 2} & \rea^{(n-2) \times (n-2)}}.
$$
Then, define the set $\cala$ as
$$
\cala = \calc \times \{x_{\ell}^*  \} \subset \rea^n,
$$
where $x_{\ell}^*   \in \rea^{n-2}$ is a constant vector and $\calc$ is a Jordan curve, given in implicit form as
\begequ
\label{calc}
\calc:= \{x_p \in \rea^2 ~|~ \Phi(x_p) =0 \},
\endequ
{on which $\nabla \Phi\neq 0$ with a smooth function $\Phi(x_p) \in \rea$.}

\begrem
\lab{rem1}
It is important at this point to clarify the difference between our objective of {\em orbital} stabilization and the more classical {\em set} stabilization. The latter is satisfied ensuring $\liminf \|x(t)\|_\cala=0$, but this does not ensure that the desired periodic motion is generated. Indeed, if the set
\begequ
\lab{ome}
\calo:=\{x\in \rea^n\;|\;\|x\|_\cala=0\}
\endequ
contains {\em equilibrium points} of the closed-loop dynamics the periodic motion is not generated. That is, we want to ensure that the closed-loop vector field \eqref{matobj} satisfies
\begequ
\lab{fneq0}
f_{\tt cl}(x)|_{x \in \calo}\neq 0.
\endequ
\endrem
%

\section{Main Results}
\label{sec4}
%
We propose in this section two methods to solve the orbital stabilization problem posed above, MSEA and EPD-PBC, whose underlying philosophy is described in Section \ref{sec2}. Connections between these two methods are also given.
\subsection{Mexican sombrero energy assignment PBC}
\label{subsec51}
The successful application of the IDA-PBC procedure described in Section \ref{sec2} is guaranteed  in MSEA with the following.

\begin{assumption}
\lab{ass1}\rm
 There are mappings \eqref{jdrd} with $\calr(x)\ge0$ and a function $H_0:\rea^{n-1} \to \rea$ verifying
 \begequ
 \label{A2}
    \arg\min\; H_0({x_0, x_\ell}) =(0, x_\ell^*  ) \quad \text{(isolated)},
 \endequ
which are solutions of the PDE \eqref{pde}, where we defined the function
$
H (x) := H_0(\Phi(x_p), x_\ell).
$
\end{assumption}

There are two additional requirements to ensure the success of the MSEA design. First, a detectability-like condition to guarantee attractivity of the desired orbit. Second, to avoid the scenario discussed in Remark \ref{rem1}, we impose a constraint on the interconnection matrix, that ensures there are no equilibrium points in the set $\calo$ given in \eqref{ome}. These requirements are articulated in the assumptions of the following proposition.
\begin{proposition}
\label{pro1}\rm
Consider the system \eqref{fgsys}, verifying Assumption \ref{ass1}, in closed-loop with the control law $u=\hat u(x)$ with $\hat u(x)$ given in \eqref{ida}. Assume the following.
\begin{enumerate}
  \item[\em H1] $\cala$ is the largest invariant set in the set
 $$
\calq:= \{x\in \rea^n | \nabla^\top H(x) \calr (x) \nabla H(x) =0 \} \cap B_\varepsilon (\cala) ,
 $$
 for some $\varepsilon>0$.
  \item[\em H2] The (1,2)-element of $\calj(x)$ may be parameterized as
\begequ
\label{param1}
 \calj_{(1,2)}(x) =  { c(x) \over \nabla_{x_0}  H_0(x_0, x_\ell)}\Bigg|_{x_0 = \Phi(x_p)}
\endequ
for some $c: \rea^n \to \rea$ satisfying
$
0 < |c(x) | < \infty, \; \forall x \in \cala.
$
\end{enumerate}

Then, the closed-loop system is asymptotically orbitally stable.
\end{proposition}
\begin{proof}\rm
The closed-loop system takes the form
\begequ
\label{pH}
\dot{x} = [\calj(x) - \calr(x)]\nabla H.
\endequ
From the isolated minimum condition of $H_0(x_0,x_\ell)$ stated in \eqref{A2}, we conclude that the function $H(x)$ has minima in the set $\cala$. Consequently,
\begequ
\label{HA}
\nabla H(x)\big|_{x \in \cala} =0, \;
{\nabla^2 H_0(x)\big|_{x \in B_\varepsilon(\cala)} >0}
\endequ
for some $\varepsilon>0$. This shows that the set $\calq$---containing the set $\cala$---is non-empty.

From the closed-loop pH dynamics it is clear that
$$
\dot{H} = - \nabla^\top H(x) \calr(x) \nabla H(x) \le 0,
$$
implying the boundedness of $H(x)$. Together with \eqref{HA}, we conclude the Lyapunov stability of the closed-loop system with respect to $\cala$. Thus, given a parameter $\varepsilon>0$, there always exists an \emph{invariant} set $\cale$ such that
$
\cala \subset \cale \subset B_\varepsilon(\cala).
$
Now, from the first equation of \eqref{HA} we get
$$
\nabla^\top H(x) \calr (x) \nabla H(x) \big|_{x \in \cala} =0.
$$
Applying LaSalle's invariance principle, taking into account the trajectory boundedness in $\cale$, and the assumption {\em H1}, we prove the attractivity of $\cala$, that is,
$$
\lim_{t\to \infty} \|x(t)\|_\cala = 0,\quad  \forall x(t_0) \in \cale.
$$

The proof is completed establishing the existence of the periodic orbit, that is, verifying \eqref{fneq0}. Consider the term $\calj_{pp}(x)\nabla_{x_p} H(x)$ of the closed-loop dynamics:
\begalis{
&\calj_{pp}(x)\nabla_{x_p} H(x) = \calj_{pp}(x) \nabla_{x_0} H_0(x_0,x_\ell) \nabla \Phi(x_p)\\
&= \begmat{0 & \calj_{(1,2)}(x) \\ - \calj_{(1,2)}(x) & 0} \nabla_{x_0} H_0(x_0,x_\ell) \nabla \Phi(x_p) \\
& =\begmat{0 & c(x) \\ - c(x) & 0} \nabla \Phi(x_p),
}
where we applied in the first identity the chain rule
$
\nabla_{x_p} H(x) = \nabla_{x_0} H_0(x_0,x_\ell) \nabla \Phi(x_p)
$, and used assumption {\em H2} in the third one. Considering that $\nabla H(x)\big|_{x \in \cala} =0$, the residual dynamics is
\begalis{
\dot{x}_p =  \begmat{0 & c(x) \\ - c(x) & 0} \nabla \Phi(x_p) , \quad
  \dot{x}_\ell  = 0.
}
Now, from $\dot{\Phi} = 0$, we conclude that the set $\calc$ is invariant. To prove that the set $\cala$ is a periodic orbit, we compute the $1$-norm of $\dot{x}_p$ as
$$
\|f_{\tt cl}(x)\|_1 = |c(x)| \|\nabla \Phi(x_p)\|_1 >0, \quad \forall x \in \cala.
$$
With the additional Jordan curve assumption, the existence of a periodic orbit is verified, completing the proof.
\qed
\end{proof}
\begin{remark}\rm
\lab{rem2}
The minimum condition \eqref{A2} implies that $\nabla_{x_0} H_0(x_0,x) =0$. Consequently, in view of condition \eqref{param1}, the mapping $\calj_{p\ell}(x)$ is singular along the orbit. However, the closed-loop dynamics and the feedback law \eqref{ida} are well-defined everywhere. If the term $\calj_{(1,2)}(x)$ is bounded along the orbit the condition \eqref{fneq0} is violated. Consequently, the ``infinite interconnection" condition {\em H2}  is {\em necessary} to ensure the orbit exists, otherwise we only achieve set stabilization---see Remark \ref{rem1}.
\end{remark}

\subsection{Energy pumping-and-damping PBC}
\label{subsec52}

In this subsection, we introduce an alternative orbital stabilization methodology: EPD-PBC---where the periodic orbit is enforced by regulating the \emph{energy level} to a constant value. More precisely, we assume the total energy of the closed-loop can be decomposed as
\begequ
\lab{hdphdl}
H(x):=H_{p}(x_p)+H_{\ell}(x_\ell).
\endequ
The function that defines the Jordan curve \eqref{calc} is given as
\begequ
\lab{phi}
\Phi(x_p):= H_{p}(x_p) - H_{p}^*,
\endequ
with $H_{p}^*$ the desired energy level for $H_{p}(x_p)$, which should be ``above" the minimal value of $H_p(x_p)$, that is, it should satisfy
$$
H_{p}^* > \min({H_p}(x_p)).
$$
To enforce the oscillation, the ``sign" of the damping matrix $\calr(x)$ changes according to the position of the state $x$ relative to the desired oscillation---whence, to the set $\calc$. See Fig. \ref{fig:pd}.

Similarly to Assumption \ref{ass1} for MSEA-PBC, in EPD-PBC we require that the PDE \eqref{pde} is solvable, with an additional constraint on $\calr(x)$ to implement the energy pumping-and-damping mechanism.

\begin{assumption}\rm
\label{ass2}
There exist functions $H_{p}(x_p)$ and $H_{\ell}(x_\ell)$, which have isolated minima in $x_p^* \in \rea^2$ and $x_\ell^*\in \rea^{n-2}$, respectively, and mappings \eqref{jdrd}, with
\begalis{
\calr(x) & =\diag\{\calr_{pp}(x),\calr_{\ell\ell}(x_\ell)\},
}
where $\calr_{\ell\ell}(x_\ell)  \geq 0$ and {the diagonal matrix} $\calr_{pp}(x)$ satisfies the pumping-and-damping condition
  \begequ
  \label{pd-prop}
   \calr_{pp}(x)\Phi(x_p)\ge 0
  \endequ
where $\Phi(x_p)$ is given in \eqref{phi}, and
  \begequ
  \label{pd-prop2}
    \calr_{pp}(x) =0 \quad \Longleftrightarrow \quad \Phi(x_p) =0.
  \endequ
\end{assumption}

In EPD-PBC besides the detectability-like and the interconection conditions, we require a technical assumption to complete the proof. That is, $\nabla^\top H_{p}(x_p)\calj_{p\ell}(x) = 0$, in order to ``cut off" the energy flow between $x_\ell$ and $x_p$ partitions.

\medskip{}

\begin{proposition}
\label{pro3}\rm

Consider the system \eqref{fgsys}, verifying Assumption \ref{ass2}, in closed-loop with the control law $u=\hat u(x)$ with $\hat u(x)$ given in \eqref{ida}. Assume the following.
\begenu
  \item[\em H3] \{$x_\ell^{*}$\} is the largest invariant set in the set
 $$
 \big\{x_\ell \in \rea^{n-2} \big| \nabla^\top  H_{\ell}(x_\ell) \calr_{{\ell\ell}}(x_\ell) \nabla H_{\ell}(x_\ell) =0 \big\}.
 $$
\item[\em H4] The matrix {$\calj(x)$} satisfies
$$
\calj_{(1,2)}(x) \neq 0,\quad \nabla^\top H_{p}(x_p)\calj_{p\ell}(x) = 0.
$$
\item[\em H5] {For some $\varepsilon_* >0$}
\begalis{
\nabla^2 H_p|_{x_p\in B_{\varepsilon_*}(x_p^*)}  >0 , \;
\max_{B_{\varepsilon_*}(x_p^*)} H_p(x_p)  > H_p^*.
}
\endenu

Then, the closed-loop system is asymptotically orbitally stable with respect to the orbit
$
{\cala \cap  B_{\varepsilon_*}(x_*)\ }.
$
\end{proposition}
\begin{proof}\rm
The closed-loop dynamics takes the form \eqref{pH} with $\nabla H=\col(\nabla H_p(x_p) , \nabla H_\ell(x_\ell))$. From which it is clear that
$$
\begin{aligned}
\dot{H}_\ell & =  - \nabla^\top H_\ell \big( \calr_{\ell\ell}(x_\ell) \nabla H_\ell + \calj^\top_{p\ell}(x) \nabla H_p \big) \\
& =  - \nabla^\top H_\ell \calr_{\ell\ell}(x_\ell) \nabla H_\ell  \le 0,
\end{aligned}
$$
where we have used the assumption {\em H4}. Applying LaSalle's invariance principle and using the assumption {\em H3}, we have
$$
\lim_{t\to\infty} x_\ell(t) = x_\ell^*.
$$
For $H_p(x_p)$ we have
$$
\begin{aligned}
\dot{H}_{p} & = - \nabla^\top H_p \big( \calr_{pp}(x) \nabla H_p - \calj^\top_{p\ell}(x) \nabla H_\ell \big)\\
& = - \nabla^\top H_p  \calr_{pp}(x) \nabla H_p,
\end{aligned}
$$
where we used again the assumption {\em H4}.

Consider the function
$
V(x_p) := {1\over 2} \Phi^2(x_p),
$
we have
$$
\begin{aligned}
   \dot{V} & = - \nabla^\top H_p [\Phi(x) \calr_{pp}(x) ]\nabla H_p \le 0,
\end{aligned}
$$
where the inequality is the consequence of the pumping-and-damping condition \eqref{pd-prop}. Invoking LaSalle's invariance principle, the state ultimately converges into the largest invariant set of the set
$$
\calq= \{x\in \rea^n | x_\ell =x_\ell^*,\; [\Phi(x) \calr_{pp}(x) ]\nabla H_p(x) =0 \}.
$$
There are three cases of $ [\Phi(x) \calr_{pp}(x) ]\nabla H_p(x) =0$, namely,
\begin{itemize}
  \item[i)] $ [\Phi(x) \calr_{pp}(x) ] =0$;
  \item[ii)] $\nabla H_p(x) =0$;
  \item[iii)] $\Phi(x) \calr_{pp}(x) \neq {\bf 0}$ and $\nabla H_p(x) \neq {\bf 0}$ with
  $$
  \nabla H_p(x) \in \mathtt{Ker}(\Phi(x) \calr_{pp}(x) )= \mathtt{Ker}(\calr_{pp}(x) ).
  $$
\end{itemize}

First consider Case iii) with the definition
$$
\mathcal{Q}_{iii} := \{x\in\rea^n | x_\ell = x_\ell^* \text{ and } x \text{ satisfies Case iii)} \}.
$$
We will prove that $\calq_{iii}$ is not an invariant set by contradiction. Assume $\calq_{iii}$ is invariant along the closed-loop dynamics. On $\calq_{iii}$ the residual dynamics is
$$
\begin{aligned}
 \dot{x}_p & = [\calj_{pp}(x) - \calr_{pp}(x)] \nabla H_p(x_p)\\
\end{aligned}
$$
with $x_\ell = x_\ell^*$. Since $\nabla H_p(x_p) \in \mathtt{Ker}(\calr_{pp}(x))$, we have
\begequ
\label{case3}
\dot{x}_p = \calj_{pp}(x)\nabla H_p.
\endequ
Assumption {\em H4} ensures $\det(\calj_{pp}(x)) \neq 0$, hence from \eqref{case3} we conclude that there are no equilibrium points in ${\calq}_{iii}$. From \eqref{case3} we also conclude that
\begequ
\label{H_pconst}
 H_p(x_p(t)) \equiv \text{const}, ~ \; \forall x\in {\calq}_{iii}, \; t\ge 0.
\endequ
 Noticing the diagonal condition of $\calr_{pp}(x)$, together with $\det(\calr_{pp})=0$, $\calr_{pp} \neq {\bf 0}$ for Case iii) and $ \nabla H_p(x) \in  \mathtt{Ker}(\calr_{pp}(x) )$, we then have
$
\nabla_{x_{p1}} H_p \equiv 0 \; \text{or} \;
\nabla_{x_{p2}} H_p \equiv 0,
$
that contradicts the identity \eqref{H_pconst}. Therefore, the set ${\calq}_{iii}$ is not invariant, excluding the possibility of Case iii).

For Case i), from \eqref{pd-prop2} we have
$$
\Phi(x) \calr_{pp}(x) = 0 \quad \Rightarrow \quad \Phi(x_p) = 0.
$$
{Together with $\dot{\Phi} =0$ for all $x\in \cala$, it implies the invariance of Case i).} For Case ii), it yields $x = \col(x_p^*, x_\ell^*):= x_*$. In summary, the largest invariant set in ${\calq}$ is $\cala\cup \{x_*\}$.

We consider the function
$
W(x) = \Phi(x),
$
and, for some small $\varepsilon>0$, it follows
$$
\dot{W} = - \nabla^\top H_p \calr_{pp}(x) \nabla H_p \ge0, \quad \forall x_p \in B_\varepsilon (x_p^*).
$$
Therefore, the isolated equilibrium point $x_*$ is unstable. On the other hand, the set $\cala$ is attractive.

We proceed now to verify the existence of a periodic orbit. Since $x_*$ is a minimum of $H(x)$ we have that
$$
\nabla H(x)\big|_{x=x_*} =0, \;
\nabla^2 H(x)\big|_{x\in B_\varepsilon(x_*)} >0.
$$
If $x\in B_\varepsilon(x_*)$, the function $\mathcal{V}(x):= H(x) - H(x_*)$ qualifies as a Lyapunov function (for the dynamics $\dot{x} = F(x)\nabla H(x)$ with $F(x)>0$). According to \cite[Theorem 4.1]{BYR}, the set $\calc\cap B_{\varepsilon_*}(x_p^*)$ defines a Jordan curve. On the set ${\cala \cap B_{\varepsilon_*}(x_*) }$ the residual dynamics is
$$
\begin{aligned}
\dot{x}_p  = \begmat{0 & \calj_{(1,2)}(x) \\ - \calj_{(1,2)}(x) & 0} \nabla H_p, \;
\dot{x}_\ell   = 0.
\end{aligned}
$$
We conclude $|f_{\tt cl}(x)| \neq 0$, completing the proof.
\qed
\end{proof}

\begin{remark}
\rm
The condition {\em H4} is similar, in nature, to {\em H2}, but excluding equilibria on the orbit ${\cala \cap B_{\varepsilon_*}(x_*) }$. Noticing that $|\nabla H_{p}| \neq 0$ on the desired orbit, only ``finite interconnection'' is adequate in the EPD method for the purpose of orbital stabilization. An example of energy regulation without adequate interconnection is given in our previous work \cite{YIetalscl}, which solves the open problem---using \emph{smooth, time-invariant} state-feedback to achieve almost global asymptotic \emph{regulation} of three-dimensional nonholonomic systems.
\end{remark}
\begin{remark}
\rm
A trivial selection of the mapping $\calr_{pp}$ is $\diag(0,\Phi(x_p))$, but it is non-unique. This indeed provides an additional degree of freedom to solve the PDE, and the possibility to regulate the  speed of convergence.
\end{remark}
\subsection{Comparison of MSEA-PBC and EPD-PBC}
\label{subsec:relation}

In this section we compare the two methods and clarify the parallel between them. To simplify the presentation we relabel the various mappings used in the methods with fonts \texttt{mathcal} ($\mathcal{J, R, H}$) for MSEA-PBC and \texttt{mathbf} ($\mathbf{J,R,H}$) for EPD-PBC. We have the following.
\begin{proposition}
\label{pro4}\rm
Consider the system \eqref{fgsys}, verifying all the assumptions in Proposition \ref{pro1}. Assume the matrix $\frakr$ is diagonal, $\frakr_{\ell\ell}$ is a function of $x_\ell$, and $\frakr_{pp}$ is non-zero. If the mapping $\frakh_0 : \rea^{n-1} \to \rea$ can be decomposed as
\begequ
\label{decomp}
\frakh_0 (x_0,x_\ell) = \frakh_{1}(x_0) + \frakh_{\ell}(x_\ell),
\endequ
and
$
\nabla^\top \Phi (x_p) \frakj_{p\ell}(x) =0,
$
then all the assumptions in Proposition \ref{pro3} are satisfied by selecting the mappings\footnote{The notation ${\frakh_{1}'}$ represents the derivative of $\frakh_{1}(x_0)$ with respect to $x_0$. We also have $\frakh_1' (x_0)= \nabla_{x_0} \frakh_0(x_0,x_\ell)$ according to \eqref{decomp}.}
\begequ
\label{conin}
\begin{aligned}
  \bfh_p(x_p) & =\Phi(x_p), \quad \quad
  \bfh_\ell (x_\ell) = \frakh_\ell(x_\ell) \\
 \bfj(x)  &=  \begin{bmatrix}  {\frakh_{1}'}\frakj_{pp}  & \frakj_{p\ell} \\ -\frakj_{p\ell}^\top & \frakj_{\ell\ell} \end{bmatrix},\;
 \bfr(x)  = \begin{bmatrix} {\frakh_{1}'} \frakr_{pp} & \mathbf{0} \\ \mathbf{0} & \frakr_{\ell\ell} \end{bmatrix},
\end{aligned}
\endequ
and $\bfh_{p}^*=0$. Furthermore, the MSEA and EPD methods yield the same feedback law.
\end{proposition}
\begin{proof}\rm
We first verify the solvability of the matching PDEs---equivalently the coincidence of two closed-loop dynamics. The closed-loop dynamics in Proposition \ref{pro1} is
$$
\begin{aligned}
    \dot{x} & = \big[ \frakj(x) - \frakr(x) \big] \nabla \frakh_0(\Phi(x_p), x_\ell) \\
                & = \begin{bmatrix} \frakj_{pp}(x)  - \frakr_{pp}(x) & \frakj_{p\ell}(x) \\ -\frakj^\top_{p\ell}(x) & \frakj_{\ell\ell}(x) - \frakr_{\ell\ell}(x) \end{bmatrix}
                       \begin{bmatrix}  {\frakh_{1}'} \nabla \Phi \\ \nabla \frakh_{\ell}  \end{bmatrix}\\
                & = \begin{bmatrix}   {\frakh_{1}'} (\frakj_{pp}(x) - \frakr_{pp}(x))  & \frakj_{p\ell}(x) \\ -\frakj^\top_{p\ell}(x) & \frakj_{\ell\ell}(x) - \frakr_{\ell\ell}(x) \end{bmatrix}
                        \begin{bmatrix}  \nabla \Phi \\ \nabla \frakh_{\ell}  \end{bmatrix}\\
                & = \big[ \bfj(x) - \bfr(x) \big] \nabla \bfh(x),
\end{aligned}
$$
where we have used the assumption $\nabla^\top \Phi (x_p) \frakj_{p\ell}(x) =0$ in the third equality. It is obvious that the closed-loop dynamics in Proposition \ref{pro3} is \emph{exactly} the same with the one in Proposition \ref{pro1}. The matching PDE in Proposition \ref{pro3} is thus solvable.

Second, we will verify the assumptions in Proposition \ref{pro3}. With the decomposition \eqref{decomp}, we have
\begequ
\label{decomp2}
\eqref{A2}\quad \Longleftrightarrow \quad
\left\{
\begin{aligned}
\; \arg\min ~ \frakh_{1}(x_0) & = 0\\
\; \arg\min ~ \frakh_{\ell} (x_\ell) & = x_\ell^*,
\end{aligned}
\right.
\endequ
satisfying the convex properties of $\bfh_\ell(x_\ell)$ in Proposition \ref{pro3}. We then need to prove that there exists a point $x_p^*$ such that
\begequ
\label{Hp_prop}
\nabla \bfh_{p}(x_p^*)=0, \quad \nabla^2 \bfh_{p}(x) \big|_{x\in B_\varepsilon(x_p^*)}>0.
\endequ
with $\bfh_{p}(x_p)= \Phi(x_p)$. To this end, we notice that $\calc$ is diffeomorphic to the unit circle, and thus there exists a smooth mapping $T: \rea^2 \to \rea^2$ such that
$
\Phi(x_p) = |T(x_p)|^2 -1,
$
with $\nabla T \neq 0$ and its inverse mapping $T^{-1} (\cdot)$ is well-defined. By fixing $x_p^* = T^{-1}(\mathbf{0})$, we then have
$$
\begin{aligned}
\nabla \bfh_p\Big|_{x=x_p^*} & = 2 (\nabla T)^\top  T(T^{-1}(\mathbf{0})) =0\\
\nabla^2 \bfh_p \Big|_{x\in B_\varepsilon (x_p^*)} & =  2(\nabla T)^\top \nabla T + 2 \sum_{i=1}^{2}T_i(x_p)\nabla^2 T_i   >0
\end{aligned}
$$
for some small $\varepsilon>0$, where in the latter inequality we have used the continuity of the mapping and the fact $T_i(x_p^*)=0$. Thus we have verified the property of the Hamiltonian $\bfh(x)$ in Proposition \ref{pro3}.

To verify the condition \eqref{pd-prop}, we have
$$
\begin{aligned}
\bfr_{pp}(x) =  0 \quad & \Longleftrightarrow \quad  {\frakh_{1}'}\frakr_{pp}(x)=0 \\
                                  & \Longleftrightarrow \quad   {\frakh_{1}'}(x_0) \Big|_{x_0=\Phi(x_p)} =0 \\
                                  & \Longleftrightarrow \quad  \Phi(x_p) =0,
\end{aligned}
$$
where we have used the assumption $\frakr_{pp}(x) \neq \mathbf{0}$ in the second implication, and the fact ${\frakh_1'}|_{x_0=0}= 0$ in the last one. Therefore, the equation \eqref{pd-prop2} is satisfied. According to the property of $\frakh_1$, for sufficiently small $|x_0|$ we have ${ \frakh_{1}'(x_0) }<0$ for $x_0<0$ and ${\frakh_{1}'(x_0) }>0$ if $x_0>0$. It yields
$$
\bfr_{pp}(x) \Phi(x_p) = \frakr_{pp}(x)  \big[{\frakh_1'}(\Phi(x_p))\Phi(x_p)\big] \ge 0.
$$
Thus, the pumping-and-damping condition---inequality \eqref{pd-prop}---has been proved. The remaining assumptions {\em H3} and {\em H4} are trivially verified.
\qed
\end{proof}

\subsection{Discussions}
\label{sec44}

The following remarks are in order specifically emphasizing the connections between the proposed methods and existing methods.

\begin{remark}
\rm
In the problem formulation, we impose the two-dimensional partition of $x_p$. It may be argued to be peculiar and stringent. We should underscore that, in many cases, orbital stabilization tasks can be translated in to our case. We take the widely studied VHC method, though with a different mechanism from the proposed designs, for instance, and consider an Euler-Lagrange system with states $q\in\rea^N$ and $\dot{q}\in\rea^N$ the generalized coordinates and velocities. The simplified control task in VHC is to stabilize the invariant manifold
$$
\calm := \{ (q,\dot{q}) ~|~ \bar{q} = \alpha(q_N) \}, \quad \bar{q}:= \col(q_1, \ldots,q_{N-1}),
$$
with some mapping $\alpha: \rea \to \rea^{N-1}$, and guarantee the zero dynamics to admit non-trivial periodic solutions. It is clear that on the manifold we have
$
\dot{\bar{q}} = \eta(q_N, \dot{q}_N) := \nabla \alpha^\top(q_N) \dot{q}_N.
$
The above-mentioned task appropriately adopts to our problem formulation with $n=2N$ and the change of coordinate $x_p = \col(q_N, \dot{q}_N), \; x_\ell = \col(\bar{q} - \alpha(q_N),  \dot{\bar{q}} - \eta(q_N, \dot{q}_N)) $.
\end{remark}
\begin{remark}
\em
A main drawback of VHC and I\&I orbital stabilization technique is that the steady-state behavior cannot be fixed a priori, but depends on the initial states with a notable exception \cite{MOHetalaut}. The drawback can be circumvented with the IDA methods, but with an additional difficulty in solving PDEs.
\end{remark}

\begin{remark}
\rm\lab{rem3}
In \cite{ARCetaltac,DUISTRejc,GOMetalijrnlc} a similar MSEA approach is adopted for some specific dynamical systems. In particular, in \cite{DUISTRejc} the MSEA is imposed to the potential energy in a path following task for \emph{fully actuated} mechanical systems.
\end{remark}

\begin{remark}
\rm
In \cite{ASTetalauto} the pumping-and-damping injection is applied to \emph{stabilize} pendula at the upright equilibrium almost globally. Some works on energy regulation of nonlinear systems, though not aiming at oscillation generation, can be found in \cite{FRAPOG,GAROTTtac,SPO}.
\end{remark}

\begin{remark}
\rm
In \cite{STASEPtac} passive oscillators is constructed for Lure dynamical systems using ``sign-indefinite" feedback static mappings, see Fig. \ref{fig:lure}. After assigning the linearized system with a unique pair of conjugated poles on the imaginary axis, the ``sign-indefinite" feedback is adopted to regulate the energy achieving periodic oscillations. Indeed, \cite[Theorem 2]{STASEPtac} can be regarded as an EPD controller. It should be underscored that the center manifold theory is applied in the analysis where the center manifold plays the exactly same role as the invariant manifold in VHC. {Since the analysis of the latter is carried out applying the center manifold theory---whose nature is intrinsically local---the oscillators resulting from \cite[Theorem 2]{STASEPtac} are assumed to have small amplitudes. On the other hand, they circumvent the daunting task of solving PDEs. Whereas, the proposed EPD method has the ability to shape behaviors of the closed-loop dynamics, making it instrumental in engineering practice.}
\end{remark}

\begin{figure}{}
  \centering
  \includegraphics[width=4cm]{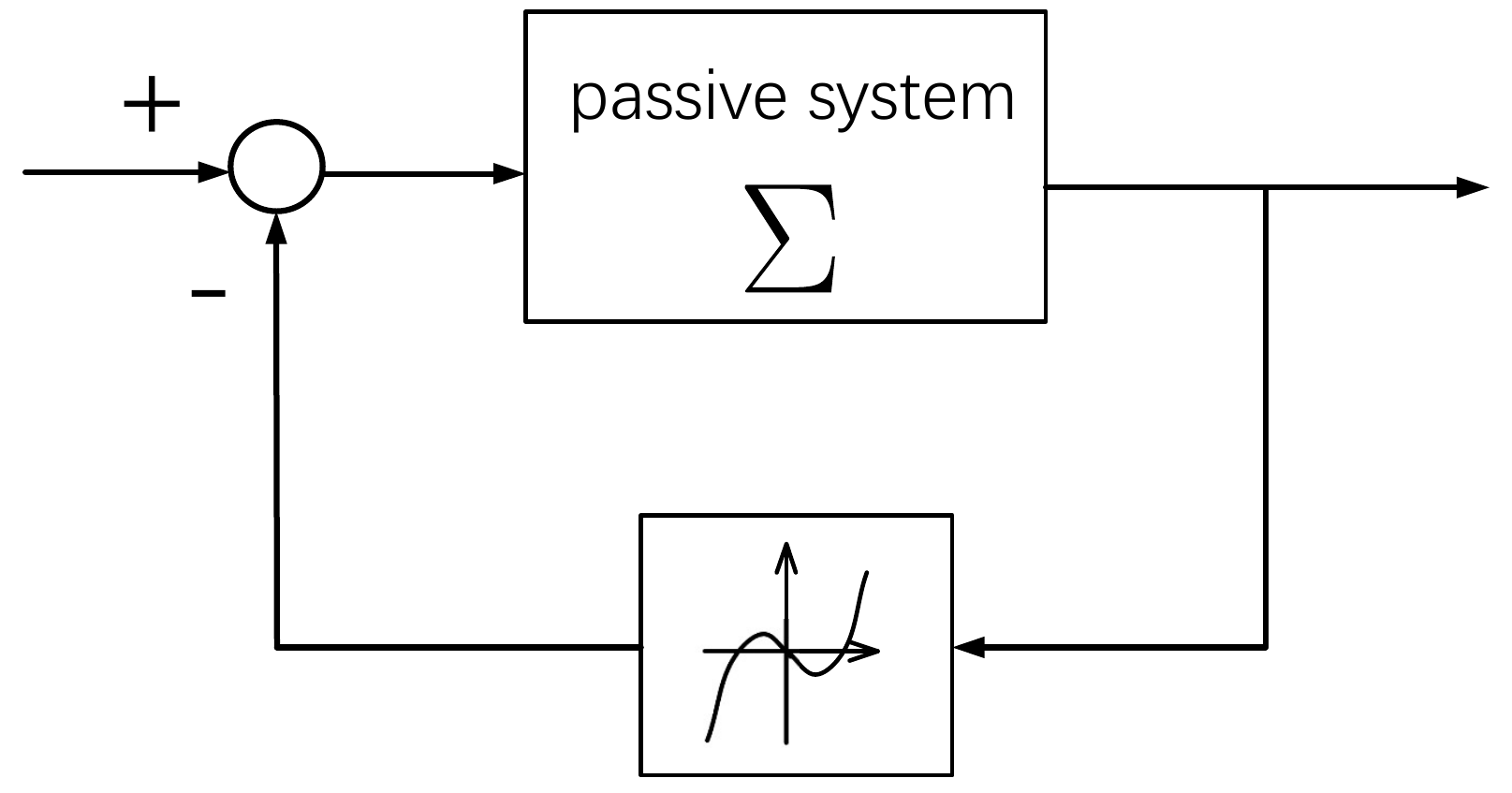}
  \caption{Lure system with sign-indefinite feedback}\label{fig:lure}
\end{figure}

\begin{remark}
\rm
{The assumptions in Proposition \ref{pro4} on the equivalence between two proposed methods are relatively mild, namely, the diagonalization of $\calr(x)$ and the decomposition of $\mathcal{H}_0(x)$. Despite the equivalence, the realms of applicability of the methods are slightly different. For instance, if a controlled plant endows a pH form, it may be easier to generate oscillations via EPD without solving PDEs; on the other hand, for some systems it is simple to shape the Hamiltonian, {\em e.g.}, fully-actuated mechanical systems, thus the MSEA method is preferred. }
\end{remark}

\section{Induction Motor Example}
\label{sec5}
\subsection{Dynamic model and control objective}
We consider the practical example of speed regulation of current-fed IMs. The normalized dynamics of the IM in the fixed frame is described by
\begequ
\label{im1}
\begin{aligned}
    \dot{\psi}_r & = -R\psi_r + \omega\bbj\psi_r + Ru \\
    \dot{\omega} & = u^\top \bbj \psi_r - \tau_L,
       \quad {\mathbb{J}}:= \begin{bmatrix} 0 & -1 \\ 1 & 0 \end{bmatrix},
\end{aligned}
\endequ
where $\psi_r \in \rea^2$ is the rotor flux, $\omega \in \rea$ is the rotor angular speed, $\tau_L \in \rea$ is the constant load torque, $R>0$ is the rotor resistance, $u\in\rea^2$ is the stator current, which is assumed to be the control and, without loss of generality, we have taken the rotor inertia to be equal to one---see \cite{ORTetalbook,MARbook} for further details. To show the basic idea, we make the assumption that $\tau_L=0$, which can be removed adding an integral in the control action \cite{ORTetalbook,MARbook}.

The control objective is to ensure the asymptotic orbital stabilization of the set
\begequ
\lab{cala1}
\cala:= \{x \in\rea^3 |\; \Phi(x_p)=0, \; x_\ell =  \omega_\star \},
\endequ
where we introduced the notation $x_p:=\psi_r$, $x_\ell:=\omega$, defined the function
\begequ
\lab{phixp}
\Phi(x_p) := |x_p | - \beta_\star,
\endequ
and $\beta_\star  > 0$, $\omega_\star   \in \rea,\;\omega_\star\neq 0$ are the desired (constant) references. {Intuitively, we may fix the desired Hamiltonian as \eqref{H0_Phi} then solving the PDE.}

\subsection{Orbital stabilization of the IM via MSEA-PBC}
In the following proposition we show that the aforementioned regulation problem of IMs can be solved via MSEA-PBC.

\begin{proposition}
\lab{pro4}
\rm
Consider the fixed-frame current-fed IM model \eqref{im1} and the target set \eqref{cala1}, \eqref{phixp}.
\begenu
\item[{\em P1}] Assumption \ref{ass1} is satisfied with the choices
\begin{equation}\label{solution_ex1}
  \begin{aligned}
    \calr = \begin{bmatrix}
             R & 0 & 0 \\ 0 & R & 0 \\ 0 & 0 & {k \over \beta_\star} |x_p|
           \end{bmatrix},
    \calj  = \begin{bmatrix}
                      0 & - { {x_\ell} |x_p|  \over |x_p| - \beta_\star}& {k R \over \beta_\star  }{x_2 \over |x_p|}\\
                      * & 0 & - {k R \over \beta_\star  } {x_1\over |x_p|} \\
                     *& *& 0
                    \end{bmatrix}
  \end{aligned}
\end{equation}
and
\begequ
\label{H0_Phi}
H_0(x_0, x_\ell) =  {1\over 2}x_0^2 + {1\over 2} ( x_\ell - \omega_\star )^2.
\endequ
\item[{\em P2}] The controller \eqref{ida} takes the form
\begequ
\label{FOC2}
\begin{aligned}
u & =
\Big[ \beta_\star  I_2 -   {k \over \beta_\star}(x_\ell  - \omega_\star)\bbj\Big]  {x_p \over |x_p|},\;k>0.
 \end{aligned}
\endequ
\item[{\em P3}] All the assumptions of Proposition \ref{pro1} are satisfied.
\endenu

Consequently, the closed-loop system is asymptotically orbitally stable with respect to \eqref{cala1}.  Moreover,  the convergence is exponential.

\end{proposition}
\begin{proof}\rm
The fact that Assumption \ref{ass1} is satisfied with \eqref{solution_ex1}-\eqref{FOC2} is easily verified. Assumption {\em H2} is also satisfied, with
$
c(x) = - x_\ell |x_p|,
$
which evaluated in $\cala$ yields $- \beta_\star \omega_\star \neq 0$. The largest invariant set in
 $
\{x\in \rea^3 | \nabla^\top H(x) \calr (x) \nabla H(x) =0 \}
 $
is $\cala\cup \{\mathbf{0}\}$. Some basic Lyapunov analysis shows that the origin is an unstable equilibrium. According to Proposition \ref{pro1}, we conclude that the closed-loop system is almost globally asymptotically orbitally stable.

To establish the exponential orbital stability claim we refer to \cite{HAUCHU}, where it is shown to be equivalent to prove that the transverse coordinate
$
z:= \col(\Phi(x_p), x_\ell-\omega_\star)
$
 exponentially converges to $(0,0)$. The proof of Proposition \ref{pro1} shows that we can always find some invariant compact sets containing $\cala$. In these compact sets, $|x_p| \ge c_1$ for some $c_1>0$. Thus in the neighborhood of $\cala$, we have
$$
 \dot{z}_2 = -{k \over \beta_\star} |x_p| z_2,
$$
then yielding the exponential convergence of $z_2$ to zero. Now we have
$$
\dot{z}_1 = - R |\nabla \Phi|^2 z_1 + \et,
$$
where $\et$ is an exponentially decaying term caused by $z_2(0)$. The Jordan curve $\Phi(x)=0$ implies $\nabla \Phi \neq 0$ in the neighborhood of $\cala$, from which we conclude the exponential stability of the transverse coordinate $z$.
\qed
\end{proof}

\begrem
It can also be shown that the closed-loop system takes the pH form \eqref{pH} with
\begin{equation}\label{solution_pde}
  \begin{aligned}
    \bfr(x)  &= \begin{bmatrix}
             R(|x_p| - \beta_\star) & 0 & 0 \\ 0 & R(|x_p|- \beta_\star) & 0 \\ 0 & 0 & {k \over \beta_\star} |x_p|
           \end{bmatrix}
    \\
  \bfj(x) & =
  \begmat{        0 & - { {\omega } }|x_p| & {k R \over \beta_\star} {x_2\over |x_p|} \\
                      * & 0 & - {k R \over \beta_\star} {x_1\over |x_p|} \\
                      *& *& 0}
  \\
  \bfh(x) & = {1\over2} |x_p|^2 + {1\over 2}(\omega-\omega_\star)^2, \quad
  \bfh_p^\star  = {1\over 2}\beta_\star^2.
  \end{aligned}
\end{equation}
and satisfies all the assumptions of Proposition \ref{pro3}. Hence, the IM can be orbitally stabilized with the EPD-PBC also.
\endrem

\subsection{FOC of the IM is an MSEA-PBC}
In this subsection we prove that the MSEA controller of Proposition \ref{pro4} exactly coincides with the industry standard direct FOC first proposed in \cite{BLA}---see also \cite[Chapter 2.2]{MARbook} and \cite[Chapter 11.2.1]{ORTetalbook}.

\begin{corollary}\rm
\lab{cor1}
The MSEA controller \eqref{FOC2} of Proposition \ref{pro4} yields, after a state and input change of coordinates, the classical direct FOC.
\end{corollary}
\begin{proof}\rm
To prove that \eqref{FOC2} coincides---modulo a coordinate change---with the direct FOC we introduce the change of coordinates
\begequ
\lab{coorot}
\psi_r := e^{\bbj\theta} \lambda, \quad  u:= e^{\bbj\theta} v, \quad   \dot{\theta} = \omega,
\endequ
that, applied to \eqref{im1} (with $\tau_L=0$), yields the well-known current-fed IM dynamics in the rotating frame
\begin{equation}\label{im2}
  \begin{aligned}
    \dot{\lambda} = -R\lambda + Rv , \quad
    \dot{\omega}  = {v^\top \bbj \lambda}.
  \end{aligned}
\end{equation}
Now, we write \eqref{im2} in polar coordinates $(\beta,\rho)$ as
\begalis{
\dot \beta &= -R \beta + Ri_d,\;\dot \rho = {R \over \beta}i_q,\;\dot \omega =\beta i_q
}
where we have defined
\begalis{
\lambda := \beta \begmat{\cos\rho \\ \sin\rho },\;\begmat{i_d \\ i_q}:= e^{-\bbj \rho} v.
}
It is easy to see from the equations above that the objective $\beta(t)\to \beta_\star,\;\omega(t)\to \omega_\star$, which is equivalent to the asymptotic stabilization of the set\footnote{Notice that $|\psi_r|=|\lambda|=\beta$.} $\cala$, is achieved with the simple control
\begin{equation}
\label{foc}
v=  e^{\bbj \rho}
\begin{bmatrix}
\beta_\star   \\ {k \over \beta_\star  } \big(\omega_\star   - \omega \big)
\end{bmatrix},
\end{equation}
with $k>0$. This is the famous direct FOC for induction motors. It is a simple exercise to show that \eqref{FOC2} is obtained applying to \eqref{foc} the change of coordinates \eqref{coorot}.
\qed
\end{proof}

\begrem
It is interesting to note that, expressed in the rotating coordinates, the direct FOC does not generate a periodic orbit, but only ensures set stability.
\endrem

\begrem
The application of the main idea of FOC of IMs for smooth regulation of Brockett's non-holonomic integrator was first reported in \cite{ESCetalaut}, and later adopted in \cite{DIXetal,MORSAM} for control of nonholonomic systems.
\endrem

%
\section{Pendulum Example}
\label{sec6}
%

\subsection{Local design}
We consider a benchmark in nonlinear control---the planar inverted pendulum, which is related to various applications, {\em e.g.}, the attitude control of space boosters and walking robots. The normalized model given by \cite{ASTetalauto}
\begequ
\label{dyn_pendulum}
\begin{aligned}
\dot{\theta}  = \omega , \quad
\dot{\omega}  = \sin \theta - u \cos \theta,
\end{aligned}
\endequ
where $\theta \in \mathbb{S}$ and $\omega \in\rea$ denote the angular position and velocity, and the input $u$ is the acceleration of the pivot. In this representation, the angles $0$ and $\pi$ correspond to the upright and downright positions, respectively.

We are interested in asymptotically stabilizing the pendulum oscillating around its upright equilibrium. We define $x= x_p:= \col(\theta,\omega)$ in the absence of the $x_\ell$ partition. The first design is a local result as follows.

\begin{proposition}
\label{prop:pendulum1}\rm
Consider the model \eqref{dyn_pendulum} in closed-loop with the control law
\begequ
\label{pd_pendulum2}
    u = {2\sin \theta} +  \omega P(\theta,\omega)\cos\theta
\endequ
with
$
{1\over \gamma }P(\theta,\omega) =  - ( \cos\theta - {1\over 2})^2 +  {1\over2} \omega^2 - H_p^*,
$
where $ H_{p}^*  := -( \cos\theta_* - {1\over 2})^2$, $\gamma >0$, and $\theta_* \in (-{\pi\over 3}, {\pi\over3})$, the system is \emph{locally} asymptotically orbitally stable. Furthermore, the angle $\theta$ ultimately oscillates between $[-\theta_*,\theta_*]$.
\end{proposition}
\begin{proof}\rm
\label{proof:6}
We first define
$$
\Phi(x)= H_p - H_p^*.
$$
The closed loop takes the pH form \eqref{pH} with\footnote{The Hamiltonian function is motivated by \cite{ASTetalauto}.}
$$
\begin{aligned}
\calj(x) & =  \begmat{0 & 1 \\ -1 & 0} , \;
\calr(x)  = \begmat{0 & 0 \\ 0 & \gamma(\cos\theta)^2 \Phi(x) } \\
H_p(x) & = - \Big( \cos\theta - {1\over 2}\Big)^2 +  {1\over2} \omega^2 .
\end{aligned}
$$
The Hamiltonian function $H_p(x)$ admits an isolated \emph{local} minimal point at $(0,0)$. Define
$
\cala =\{x\in \rea^2| \Phi(x) =0 \},
$
and the periodic orbit is
$$
{\cala \cap \{x\in\rea^2| |x_1| < {\pi\over 3}\}},
$$
which is a Jordan curve.

The fact
$$
\dot{\Phi} = - \omega^2 (\cos\theta)^2 \Phi.
$$
implies that the set
$$
\Omega=\{x\in\rea^2 | |\Phi(x)|< \varepsilon, \; |x_1| < {\pi \over 3} \}
$$
is invariant for some $\varepsilon>0$. It is easy to verify the assumptions {\em H3}-{\em H5} in the set $\Omega$. We complete the proof by applying Proposition \ref{pro3}.
\qed
\end{proof}

We underscore that the level set $\Phi(x)=0$ containing two disconnected parts. Hence, the restriction $|x_1|< {1\over3}\pi$ is indispensable. {We give the simulation results in Fig. \ref{fig:example3.3} with initial values $(0.1\pi,0)$ and $(0.3\pi,0)$, where $\gamma=5$ and $H_{p*} =  -0.0429$. In this figure we show the evaluation of the (2,2)-element of $\calr(x)$, illustrating the pumping-and-damping mechanism.}

\begin{figure}[h]
  \centering
  \includegraphics[width=8.5cm]{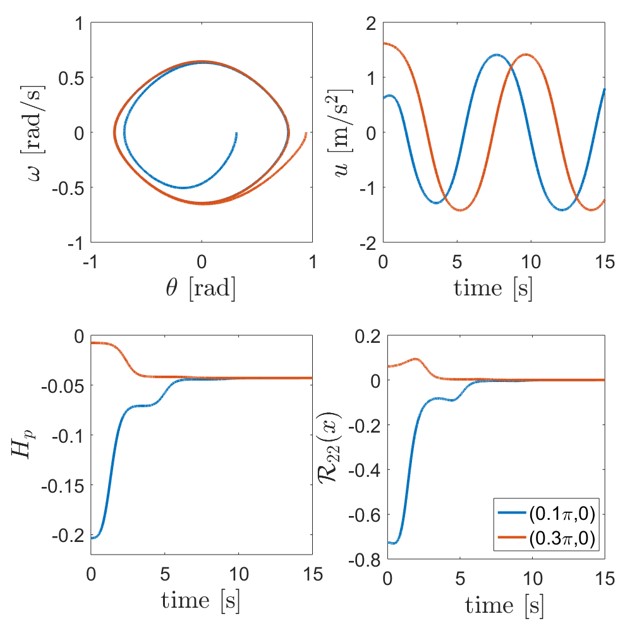}
  \caption{The dynamics behaviour in Proposition \ref{prop:pendulum1}.}\label{fig:example3.3}
\end{figure}

\subsection{Almost global design}

The following proposition is an almost global design.

\begin{proposition}
\label{prop:pendulum2}\rm
Considering Proposition \ref{prop:pendulum1}, if we select
$$
P(\theta,\omega) = ({3\over 2} \cos \theta  + {1\over2} \omega^2 -{3\over4} )  Q(\theta,\omega)
$$
and
$$
\begin{aligned}
 Q(\theta,\omega) & := \left\{
    \begin{aligned}
    & \gamma_1 (H_p(\theta,\omega) - H_{p}^*), & \theta \in (-{\pi\over3}, {\pi\over3})    \\
    & \gamma_2, & \theta \in [-\pi, {\pi\over3}] \cup [{\pi\over3}, \pi)
    \end{aligned}
 \right. \\
 \end{aligned}
$$
then there exist $\gamma_1, \gamma_2 \in \rea_+$ such that the state asymptotically converges to the orbit ${\cala \cap \{x\in\rea^2| |x_1| < {\pi\over 3}\}}$ or the saddles $({\pi\over3}, 0)$, almost globally on $\mathbb{S}\times \mathbb{R}$.
\end{proposition}
\begin{proof}
\rm
The closed-loop is the same with the one in Proposition \ref{prop:pendulum1}, but with a different function $P(\theta,\omega)$. We draw the curves of $P(\theta,\omega)=0$ (the red one) and $H_p(\theta,\omega)=0$ (the green one) in Fig. \ref{fig:example3.3}. They divide the manifold $\mathbb{S}\times \mathbb{R}$ into 9 portions, where each set is defined as an \emph{open} set. Intuitively, in the set
$$
\Omega_p := \cl(S)\cup D_1 \cup \cl(U_1) \cup \cl(U_2) \cup \cl(U_3) \cup \cl(U_4)
$$
the pumping-and-damping matrix is a pumping matrix, and in the set
$$
\Omega_d := Q_1 \cup Q_2 \cup D_2
$$
it acts as a damping one.

\begin{figure}[h]
  \centering
  \includegraphics[width=6cm]{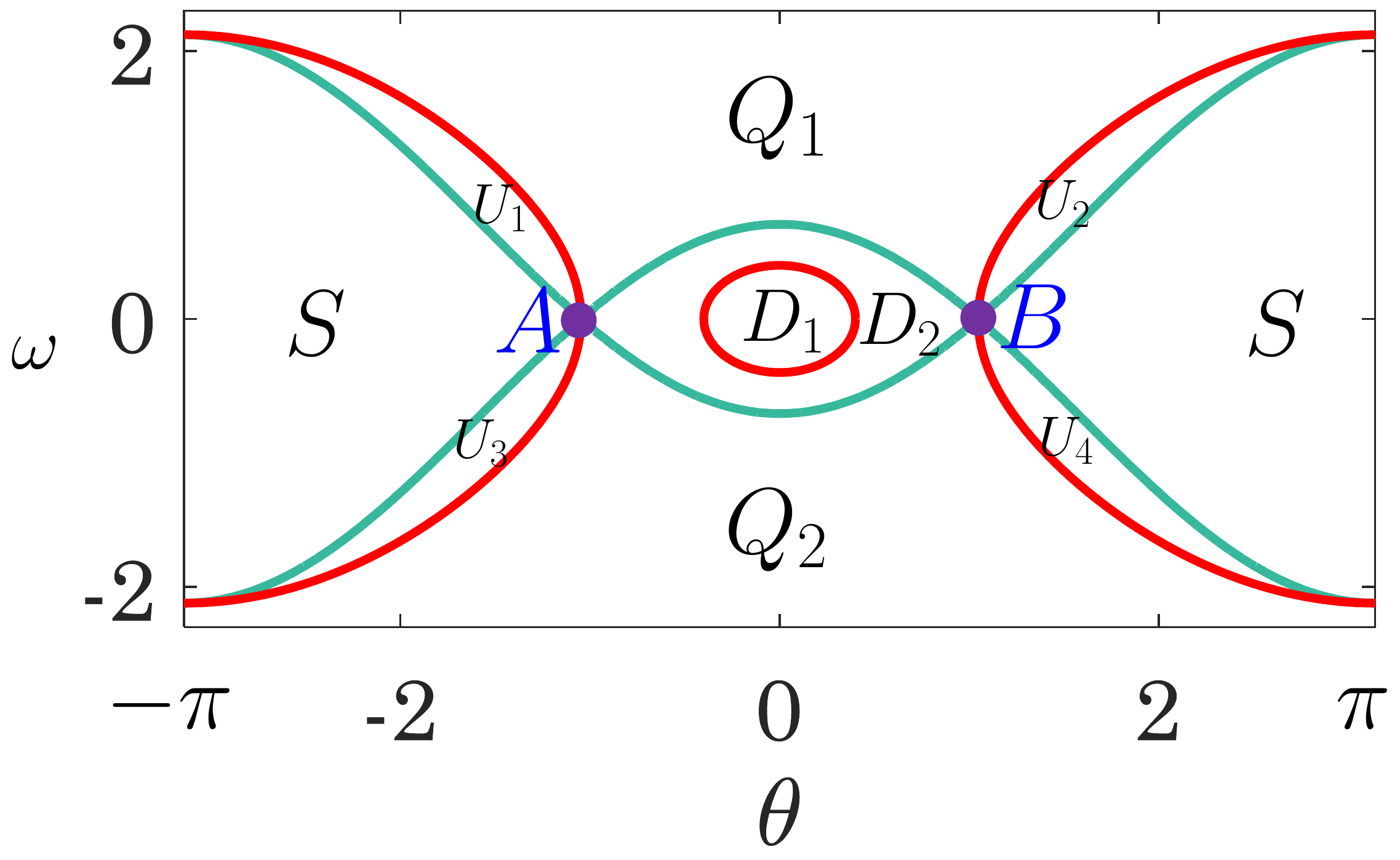}
  \caption{The partition of the space $\mathbb{S}\times \mathbb{R}$.}\label{fig:example3.3}
\end{figure}

It is clear that
$$
\dot{H}_p = (\nabla H_p)^\top
\begin{bmatrix} 0 & 0 \\ 0 & -(\cos \theta)^2 P(\theta,\omega)\end{bmatrix} \nabla H_p.
$$
We now consider three possible cases of the initial condition.

1) In the set
$$
\Omega_p - D_1 /\{(\pi,0), A,B\}
$$
we have
$
\dot{H}_p >0,
$
thus this \emph{connected} set is a repeller. Noticing that the boundary of the set $S$ is the contour $H_p(x)=0$, thus all states in the set $\Omega_p - D_1 /\{(\pi,0), A,B\}$ will leave into $\cl(U_1) \cup \cl(U_2) \cup \cl(U_3) \cup \cl(U_4)$, which contains two equilibria $A$ and $B$.

2) It is also easy to show the set
$$
\cl(D_2 )\cup \cl(D_1)/\{A,B\}
$$
is invariant. In this set, we need to prove that the control law regulates the state to the desired energy level $H_p= H_{p}^*$\footnote{It should be noticed that the contour $\Phi(x)=0$ has two unconnected parts.}. To this end, we define the function
$$
V = {1\over 2} \Phi(x)^2.
$$
Its derivative along the closed-loop dynamics is
$$
\begin{aligned}
    \dot{V} 
                 & = -(\nabla H_p)^\top  \begin{bmatrix} 0 & \\ & \calz(\theta,\omega) \end{bmatrix}\nabla H_p
\end{aligned},
$$
with 
$$
\calz(\theta,\omega) = \bigg({3\over 2} \cos \theta  + {1\over2} \omega^2 -{3\over4} \bigg) \cdot \big( \cos\theta \Phi\big)^2,
$$
which is positive semi-definite in $\cl(D_2 )\cup \cl(D_1)/\{A,B\}$. We have the set $\{(\theta,\omega)| \dot{V}=0\}$ equal to $\{(\theta,\omega)| H_p(\theta,\omega)=H_{p\star}\}\cup (0,0)$, and the equilibrium $(0,0)$ is unstable. Thus invoking LaSalle's invariance principle, if the initial state is in $\cl(D_2 )\cup \cl(D_1)/\{A,B, (0,0)\}$, it will asymptotically converge to the desired periodic orbit.

3) For the set $ \cl(U_1) \cup \cl(U_2) \cup Q_1 $ or the set $\cl(U_3) \cup \cl(U_4)  \cup Q_2$, it is relatively complicated. For convenience, we define the set
$$
\Omega_3 := Q_1 \cup Q_2 \cup \cl(U_1 \cup U_2 \cup U_3 \cup U_4).
$$
Since the set $\cl(S)$ is a repeller, the states which star in $\Omega_3$ have two possible trajectories:
\begin{description}
  \item[3a)] converging to the compact set $\cl(D_2)$, and then it can be analyzed as case 2).
  \item[3b)] staying in the set $\Omega_3$ for all $t>0$.
\end{description}
Following the proof in \cite{ASTetalauto} with some complicated analysis, we can prove the energy dissipation in $\Omega_3$ thus ruling out the case 3b) except two equilibria $A$ and $B$.
It completes the proof.
\qed
\end{proof}

Fig. \ref{fig:example3.4} gives the simulation results to illustrate Proposition \ref{prop:pendulum2} with the $x(0)= (\pi, 0.01)$, $\gamma_1=20,\; \gamma_2=2, \theta_* = {\pi \over 4}$ and {$H_{p*}=-0.0429$}. The figure illustrates the almost global property with the angle ultimately oscillating between $[-{\pi\over4}, {\pi\over4}]$.

\begin{figure}[h]
  \centering
  \includegraphics[width=8.5cm]{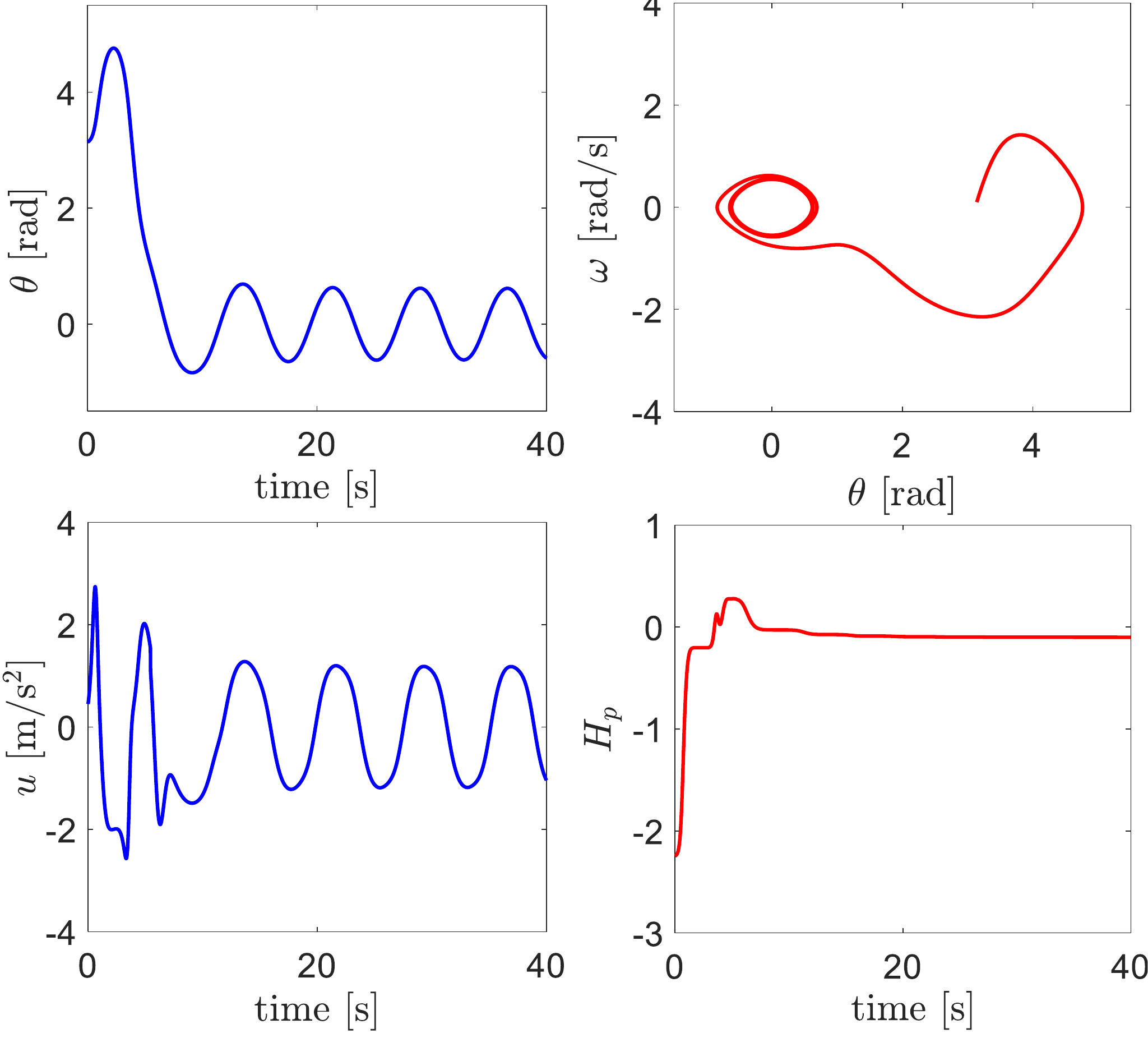}
  \caption{The dynamics behaviour nin Proposition \ref{prop:pendulum2}.}\label{fig:example3.4}
\end{figure}

\section{Concluding Remarks}
\label{sec7}
It has been shown that the IDA-PBC design methodology can be adapted to address the problem of orbital stabilization of nonlinear systems. We propose two different, but related, IDA-PBC designs: MSEA and EPD---whose application, as usual in IDA, requires the solution of a PDE. In the former,  the closed-loop Hamiltonian function is shaped to have minima at the desired orbit. For the latter, we regulate the energy to a desired value using a pumping-and-damping dissipation matrix. To ensure asymptotic orbital stability, and not just set attractivity, some constraints are imposed on the interconnection matrix. We then establish connections between the above-mentioned methods.

Currently research is carried out in the following directions.
\begin{itemize}
  \item The problem of path following---in a time parameterization-free manner---for mechanical systems. It has been observed that this is closely related to the orbital stabilization problem studied in this paper.
  \item The connection between the proposed method and the {indirect} version of FOC is still an open, and interesting, topic.
  \item Application of the proposed methods to solve some periodic motion control problems in mechanical and power electronic systems, {\em e.g.} in walking robots and AC (or resonant) power converters.
\end{itemize}

\begin{ack}                               
This paper is supported by the NSF of China (61473183, U1509211, 61627810), National Key R\&D Program of China (SQ2017YFGH001005), China Scholarship Council and by the Government of the Russian Federation (074U01), the Ministry of Education and Science of Russian Federation  (GOSZADANIE 2.8878.2017/8.9, grant 08-08).
\end{ack}
%


\begin{thebibliography}{aa}
%

\bibitem{ARCetaltac}
J. Aracil, F. Gordillo and E. Ponce, Stabilization of oscillations through backstepping in high-dimensional systems, \TAC, vol. 50, pp. 705-710, 2005.

\bibitem{ASTORT}
A. Astolfi and R. Ortega, Immersion and invariance: A new tool for stabilisation and adaptive control of nonlinear systems, \TAC, vol. 48, pp. 590-606, 2003.

\bibitem{ASTbook}
A. Astolfi, D. Karagiannis and R. Ortega, {\em Nonlinear and Adaptive Control with Applications}, Springer-Verlag, Berlin, Communications and Control Engineering, 2008.

\bibitem{ASTetalauto}
K.J. Astrom, J. Aracil and F. Gordillo, A family of smooth controllers for swinging up a pendulum, {\em Automatica}, vol. 44, pp. 1841-1848, 2008.

\bibitem{BLA}
F. Blaschke, The principle of field orientation as applied to the new TRANSVEKTOR closed loop control system for rotating field machines, {\em Siemens Review}, vol. 39, pp. 217-220, 1972.

\bibitem{BYR}
C.I. Byrnes, On Brockett's necessary condition for stabilizability and the topology of Liapunov functions on  $\rea^ n$, {\em Communications in Information and Systems}, vol 8, pp. 333-352, 2008.

\bibitem{CHE}
D. Cheban, {\em Global Attractors of Non-Autonomous Dissipative Dynamical Systems}, , World Scientifc Publishing Co. Pte. Ltd., Singapore, 2004.

\bibitem{DUISTRejc}
V. Duindam and S. Stramigioli, Port-based asymptotic curve tracking for mechanical systems, \EJC, vol. 10, pp. 411-420, 2004.

\bibitem{DIXetal}
W.E.~Dixon, D.M.~Dawson, E.~Zergeroglu and F.~Zhang.
\newblock Robust tracking and regulation control for mobile robots.
\newblock \IJRNLC, 10:\penalty0 199--216, 2000.

\bibitem{ESCetalaut}
G. Escobar, R. Ortega and M. Reyhanoglu, Regulation and tracking of the nonholonomic double integrator: A field-oriented control approach, {\em Automatica}, vol. 34, pp. 125-131, 1998.

\bibitem{FRAPOG}
A.L. Fradkov, A.Y. Pogromsky, {\em Introduction to Control of Oscillations and Chaos}, World Scientifc Publishing Co. Pte. Ltd., Singapore, 1998.

\bibitem{GAROTTtac}
G. Garofalo and C. Ott, Energy based limit cycle control of elastically actuated robots, \TAC, vol. 62, pp. 2490-2497, 2017.

\bibitem{GOMetalijrnlc}
F. Gomez-Estern, A. Barreiro, J. Aracil and F. Gordillo, Robust generation of almost-periodic oscillations in a class of nonlinear systems, \IJRNLC, vol. 16, pp. 863-890, 2006.

\bibitem{GUKHOL}
J. Guckenheimer and P. Holmes, {\em Nonlinear Oscillations, Dynamical Systems, and Bifurcations of Vector Fields}, Springer, NY, 1983.

\bibitem{HAUCHU}
J. Hauser and C.C. Chung, Converse Lyapunov functions for exponentially stable periodic orbits, \SCL, vol. 23, pp. 27-34, 1994.

\bibitem{KHA}
K.H. Khalil, {\em Nonlinear Systems}, Prentice-Hall, NJ, 3rd ed., 2002.

\bibitem{MAGCONtac}
M. Maggiore and L. Consolini, Virtual holonomic constraints for Euler-Lagrange systems, \TAC, vol. 58, pp. 1001-1008, 2013.

\bibitem{MARbook}
R. Marino, P. Tomei and C. Verrelli, {\em  Induction Motor Control Design}, Springer Verlag, London, 2010.

\bibitem{MOHetalaut}
A. Mohammadi, M. Maggiore and L. Consolini, Dynamic virtual holonomic constraints for stabilization of closed orbits in underactuated mechanical systems, \AUT, vol. 94, pp. 112-124, 2018.

\bibitem{MORSAM}
P.~Morin and C.~Samson.
\newblock Practical stabilization of driftless systems on Lie groups: The transverse function approach.
\newblock \TAC, 48:\penalty0 1496--1508, 2003.

\bibitem{ORTetalbook}
R. Ortega, A. Loria, P. J. Nicklasson and H. Sira-Ramirez, {\em Passivity-Based Control of  Euler-Lagrange Systems}, Springer-Verlag, Berlin, Communications and Control Engineering, 1998.

\bibitem{ORTetaltac}
R. Ortega, M. Spong, F. Gomez and G. Blankenstein, Stabilization of underactuated mechanical systems via interconnection and
damping assignment, {\em IEEE Transactions Automatic Control}, vol. AC-47, no. 8, pp. 1218-1233, 2002.

\bibitem{ORTetalaut}
R. Ortega, A.J. van der Schaft, B. Maschke and G. Escobar, Interconnection and damping assignment passivity-based control of port-controlled Hamiltonian systems, {\em Automatica}, vol. 38, pp. 585-596, 2002.

\bibitem{ORTGAR}
R. Ortega and E. Garcia-Canseco,  Interconnection and damping assignment passivity-based control: A survey, {\it European J of Control}, vol. 10, pp. 432-450, 2004.

\bibitem{ORTijrnlc}
R. Ortega, B. Yi, J.G. Romero and A. Astolfi, Orbital stabilisation of nonlinear systems via the immersion and invariance technique, \IJRNLC, submitted, 2018. {\texttt{\blue{arXiv:1810.00601}}}

\bibitem{STASEPtac}
G.-B. Stan, R. Sepulchre, Analysis of interconnection oscillators by dissipativity theory, \TAC, vol. 52, pp. 256-270, 2007.

\bibitem{SHIetaltac}
A.S. Shiriaev, J.W. Perram and C. Canudas-de-Wit, Constructive tool for orbital stabilization of underactuated nonlinear systems: Virtual constraints approach, \TAC, vol. 50, pp. 1164-1175, 2005.

\bibitem{SHIetaltac10}
A.S. Shiriaev, L.B. Freidovich and S.V. Gusev, Transverse linearization for controlled mechanical systems with several passive degrees of freedom, \TAC, vol. 55, pp. 893-906, 2010.

\bibitem{SPO}
M.W. Spong, The swing up control problem for the Acrobot, \CSM, vol. 15, pp. 49-55, 1995.


\bibitem{YIetalscl}
B. Yi, R. Ortega and W. Zhang, Smooth, time-varying regulation of nonholonomic systems via energy pumping-and-damping, \SCL, submitted, 2019. {\texttt{\blue{arXiv:1812.11538}}}

\bibitem{YIetalcdc}
{B. Yi, R. Ortega, D. Wu and W. Zhang, Two constructive solutions to orbital stabilization of nonlinear systems via passivity-based control, \CDC, submitted, 2019.}

\end{thebibliography}
\end{document}